\documentclass[10pt]{svproc}

\usepackage[T1]{fontenc}
\usepackage[utf8]{inputenc}
\usepackage[english]{babel}

\DeclareUnicodeCharacter{22A2}{$\vdash$}
\input{macros.sty}
\sfcommand{CIC}
\let\theCIC\CIC%

\usepackage{amsfonts}
\usepackage{amsmath}
\usepackage{amssymb}
\usepackage{stmaryrd}
\usepackage{proof}
\usepackage{graphicx}
\usepackage{xspace}
\usepackage{csquotes}
\usepackage{tikz-cd}
\usepackage{mathpartir}

\usepackage[hidelinks]{hyperref}
\usepackage[countsame]{coqtheorem} 
\usepackage{cleveref}

\newcommand{\scriptto}{\to}

\newcommand{\st}{such that\xspace}

\crefname{lemmaAux}{Lemma}{Lemmas}
\crefname{theoremAux}{Theorem}{Theorems}
\crefname{definitionAux}{Definition}{Definitions}
\crefname{factAux}{Fact}{Facts}
\crefname{corollaryAux}{Corollary}{Corollarys}

\let\psi\phi

\let\introterm\emph
\newcommand{\definedas}{:=}
\newcommand\blfootnote[1]{%
  \begingroup%
  \renewcommand\thefootnote{}\footnote{#1}%
  \addtocounter{footnote}{-1}%
  \endgroup%
}

\setBaseUrl{{https://ps.uni-saarland.de/~forster/thesis/synthetic-coq/Undecidability.}}

\title{Parametric Church's Thesis:\\Synthetic Computability without Choice}
\titlerunning{Parametric Church's Thesis: Synthetic Computability without Choice}
\author{Yannick Forster \href{https://orcid.org/0000-0002-8676-9819}{\protect\includegraphics[scale=.25]{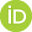}}}
\authorrunning{Yannick Forster}
\institute{Saarland University, Saarland Informatics Campus, Saarbrücken, Germany\\
\path|forster@cs.uni-saarland.de|}

\newenvironment{proofqed}{\begin{proof}}{\qed\end{proof}}
\renewcommand\L{\mathsf{L}}

\newcommand{\citetwo}[2]{\cite{#1,#2}}

\newcommand{\EAS}{\EA}
\newcommand{\SCTS}{\SCT}

\begin{document}
\maketitle
\blfootnote{\scriptsize\textcopyright~Springer Nature Switzerland AG 2022 \\
  S. Artemov and A. Nerode (Eds.): LFCS 2022, LNCS 13137, pp. 70–89, 2022. \\
  Final authenticated version: \url{https://doi.org/10.1007/978-3-030-93100-1_6}
}\vspace{-5mm}\enlargethispage{2\baselineskip}

\begin{abstract}
  In synthetic computability, pioneered by Richman, Bridges, and Bauer,
  one develops computability theory without an explicit model of computation.
  This is enabled by 
  assuming
  an axiom equivalent to postulating a function $\phi$ to be
  universal for the space $\ofbox{\nat\to\nat}$ ($\CT_\phi$, a consequence of the constructivist axiom $\CT$),
  Markov's principle,
  and at least the axiom of countable choice.
  Assuming $\CT$ and countable choice invalidates the law of excluded middle, thereby also invalidating classical intuitions prevalent in textbooks on computability.
  On the other hand, results like Rice's theorem are not provable without a form of choice.

  In contrast to existing work,
  we base our investigations in constructive type theory with a separate, impredicative universe of propositions
  where countable choice does not hold and thus a priori $\CT_{\phi}$ and the law of excluded middle seem to be consistent.
  We introduce various parametric strengthenings of $\CT_{\phi}$,
  which are equivalent to assuming $\CT_\phi$ and an $S^m_n$ operator for $\phi$ like in the $S^m_n$ theorem.
  The strengthened axioms allow developing synthetic computability theory without choice,
  as demonstrated by elegant synthetic proofs of Rice's theorem.
  Moreover, they seem to be not in conflict with classical intuitions since they are consequences of the traditional analytic form of $\CT$. %
  
  Besides explaining the novel axioms and proofs of Rice's theorem we contribute machine-checked proofs of all results in the Coq proof assistant.
\end{abstract}

\setCoqFilename{Axioms.bestaxioms}

The constructivist axiom $\CT$ (\enquote{Church's thesis})~\cite{kreisel1965mathematical,troelstra1988constructivism} states that every function $\ofbox{\nat\to\nat}$ is computable in a fixed model of computation.
In his 1992 book Odifreddi states that
the consistency of $\CT$ has been established \enquote{for all current intuitionistic systems (not involving the concept of choice sequence)}~\cite[\S 1.8 pg. 122]{odifreddi1992classical}.
For constructive type theory the consistency question is not solved entirely, but recent breakthroughs include a consistency proof for univalent type theory~\cite{swan2019church} and Martin-Löf type theory~\cite{yamada2020game}.

Assuming
$\CT$ enables developing computability theory in constructive logic in full formality without explicit encodings of programs in models of computation.
However, philosophically $\CT$ is unpleasing since its definition still requires the definition of a model.
In his seminal paper \enquote{Church's thesis without tears}~\cite{richman1983church}, Richman introduces a purely synthetic form of $\CT$ not mentioning any model of computation, which is powerful enough to develop computability theory synthetically.
The axiom is equivalent to assuming a function $\psi$ and an axiom $\CT_{\psi}$ which postulates that $\psi$ is a step-indexed interpreter universal for the function space $\ofbox{\nat\to\nat}$,
i.e.\ that every $f\of{\nat\to\nat}$ has a code $c \of \nat$ \st\ $\psi_c$ agrees with~$f$~\cite{forster2020churchs}.
Richman routinely assumes 
both Markov's principle $\MP$ and the axiom of countable choice $\AC_{\nat}$ (or even stronger forms like dependent choice),
as is the case in later work by Richman and Bridges~\cite{bridges1987varieties} and Bauer~\cite{BauerSyntCT}.
As a consequence, the law of excluded middle \LEM becomes disprovable, since $\LEM$ and $\AC_\nat$ together entail that every predicate is decidable, which is clearly in contradiction to results of synthetic computability deducible from $\CT_{\phi}$.
Text book presentations of computability however make crucial use of classical logic, rendering synthetic computability a constructivist niche.

However, Richman and Bridges state that $\AC_\nat$ can \enquote{usually be avoided}~\cite[pg.~54]{bridges1987varieties} by postulating a composition function w.r.t.\ $\phi$, which as consequence allows proving an $S^m_n$ principle w.r.t.\ $\phi$ which we abbreviate as $\SMN_\phi$.
In this paper, we show that their observation indeed holds true.

We work in the calculus of inductive constructions ($\CIC$), the type theory underlying the Coq proof assistant\rlap,\footnote{The results in the pdf version of this paper are hyperlinked with the html version of the Coq source code, which can be found at \url{https://github.com/yforster/coq-synthetic-computability/}.}
which is a foundation for constructive mathematics where the axiom of countable choice is independent, i.e.\ can be consistently assumed but is \textit{not} provable.

\paragraph{Contribution}
We 
introduce several axioms equivalent to assuming $\psi$ \st $\CT_{\psi} \land \SMN_{\psi}$.
Since working with $S^m_n$ operators in applications explicitly is tedious,
we define all axioms via
a respective notion of parametric universality. %
As a consequence, the statements of the axioms become more uniform and compact.
At the same time, the axioms become easier to use in applications.
The resulting theory enables carrying out synthetic computability theory without any form of choice axiom,
rendering the theory agnostic towards axioms like the law of excluded middle
and thereby compatible with classical intuitions.

All axioms have in common that they allow defining a enumerable but undecidable predicate $\K$, where both enumerability and decidability are defined purely in terms of functions rather than computable functions.
Since all axioms are a consequence of the constructivist axiom $\CT$,
they are consistent in \theCIC but not in contradiction to the law of excluded middle.
Thus, our axioms allow developing synthetic computability, agnostic towards classical logic.
As case studies
we give two synthetic proofs of Rice's theorem~\cite{Rice1953}:
One based on the axiom $\EPF$, following the proof approach by Scott~\cite{scott1968} relying on Rogers' fixed-point theorem~\cite{rogers1987theory},
and one based on the axiom $\EA$, establishing a many-one reduction from an undecidable problem.

\paragraph{Outline}
We motivate and introduce $\CT_{\psi}$ and $\SMN_{\psi}$ in \Cref{sec:CTdef}.
$\SCT$ is introduced in \Cref{sec:SCT}, its variants $\EPF$, $\SCT_{\bool}$, and $\EPF_{\bool}$ in \Cref{sec:EPF}.
We introduce $\EA$ in \Cref{sec:EA}
prove two synthetic versions of Rice's theorems in \Cref{sec:rice1}.

\section{Preliminaries}

We work in the Calculus of Inductive Constructions ($\CIC$), the type theory underlying the Coq proof assistant.

\subsection{Common definitions in \CIC}

We rely on the inductive types of
natural numbers $n \of \nat ::= 0 \mid \succN n$,
booleans $\bool ::= \btrue \mid \bfalse$,
lists $l \of \List X ::= [~] \mid x :: l$ where $x \of X$,
options $\option X ::= \None \mid \Some x$ where $x \of X$,
pairs $X \times Y ::= (x,y)$ where $x \of X$ and $y \of Y$,
sums $X + Y ::= \inl x \mid \inr y$ where $x \of X$ and $y \of Y$,
and,
for $p \of{ X \to \Type}$ or $p \of{ X \to \Prop}$,
$\Sigma x. p x ::= \sigpair a b $ where $a : X$ and $b : p x$.
$\pi_1$ and $\pi_2$ denote the projections $\pi_1 (a,b) := a$ and $\pi_2 (a,b) := b$.

One can easily construct a pairing function $\langle \_ \,,\, \_ \rangle \of{\nat \to \nat \to \nat}$ and for all $f \of{ \nat \to \nat \to X}$ an inverse construction $\lambda \langle  n, m \rangle .\; f n m$ of type $\nat \to X$ \st\ $(\lambda \langle n, m \rangle.\; f n m) \langle n, m \rangle = f n m$.\label{def:pairing}

For discrete $X$ (e.g.\ $\nat$, $\option \nat$, $\List \bool$, \dots), $\equivwrt{X}$ denotes equality, $\equivwrt{\Prop}$ denotes logical equivalence, $\equivwrt{A \scriptto B}$ denotes an extensional lift of $\equivwrt{B}$, $\equivwrt{A \scriptto \Prop}$ denotes extensional equivalence, and $\equivwrt{\textsf{ran}}$ denotes range equivalence.\label{introterm:equivwrt}
More formally, two functions $f, g \of X \to \option Y$ are range equivalent if $f \equivwrt{\textsf{ran}} g := \forall y.\;(\exists x.\;f x = \Some y) \leftrightarrow (\exists x.\;g x = \Some y)$.

\subsection{Partial functions}

\setCoqFilename{Shared.partial}

We work with a type $\partial A$ where $A : \Type$
for partial values
and a definedness relation $\hasvalue : \partial A \to A \to \Prop$ and write $A \pfun B$ for $A \to \partial B$.
We assume monadic structure for $\partial$ ($\retintro$ and $\bind$),\label{def:ret}
an undefined value ($\textsf{undef}$),
a minimisation operation ($\mu$), and a step-indexed evaluator ($\textsf{seval}$), see~\cite[fig.\ 2]{forster2020churchs}.

An equivalence relation on partial values can be defined as $x \equivwrt{\partial A} y := \forall a.\;x \hasvalue a \leftrightarrow y \hasvalue a$. Lifted to partial functions, we have $f \equivwrt{A \pfun B} g := \forall a b.\;f a \hasvalue b \leftrightarrow g a \hasvalue b$.

One possible definition of $\partial A$ is via stationary sequences.
We call $f \of{\nat \to \option A}$ a \emph{stationary sequence} (or just \emph{stationary}) if $\forall n a.\;f n = \Some a \to \forall m \geq n.\;f m = \Some a$.
For instance, the always undefined function $\lambda n.\;\None$ is stationary.

One can then define $\partial A := \Sigma f \of{\nat \to \option A}.\;f \textit{ is stationary}$
with 
\[ f \hasvalue a \definedas \exists n.\;\pi_1 f n = \Some a.\]

\sfcommand{mkstat}
Note that one can turn any function $f \of{\nat\to\option A}$ into a stationary sequence via a transformer $\mkstat$ with the following property.
\begin{fact}[][mkstat_spec]
  $\mkstat f \hasvalue a \leftrightarrow \exists n.\; f n = \Some a \land \forall m < n.\;f n = \None$
\end{fact}

\subsection{The universe of propositions $\Prop$, elimination, and choice principles}

$\CIC$ has a separate, impredicative universe of propositions $\Prop$ and a
hierarchy of type universes $\Type_{i}$ (where the natural number index $i$ is left out from now on).
The universe $\Prop$ is separated in the sense that the definition of functions of type $\forall x : P.\;A (x)$ for $P \of \Prop$ and $A \of P \to \Type$ by case analysis on $x$ are restricted.

In $\CIC$, both dependent pairs ($\Sigma$) and existential quantification $\exists$ can be defined using inductive types.
We verbalise $\exists x$ with \enquote{there exists $x$} and in contrast $\Sigma x$ as \enquote{one can construct $x$}.
Dependent pairs can be eliminated into arbitrary contexts, i.e.\ there is an elimination function of type
\[\forall p \of{ (\Sigma x.\; A x) \to \Type}.\; (\forall x \of X.\forall y \of {A x}. \; p \sigpair x y) \to \forall s.\; p s.\]

In contrast, existential quantification can only be eliminated for $p \of{ (\exists x.\; A x) \to \Prop}$.

This is because \theCIC forbids so-called large eliminations~\cite{paulin1993inductive} on the inductively defined $\exists$ predicate.
To avoid dealing with Coq's $\texttt{match}$-construct for eliminations in detail, we instead talk about \introterm{large elimination principles}.
A large elimination principle for $\exists$, which would have the following type, is \textit{not} definable in \theCIC:
\[\forall p \of{ (\exists x.\; A x) \to \Type}.\; (\forall x \of X.\forall y \of {A x}. \; p \sigpair x y) \to \forall s.\; p s.\]

In particular, this means that one \textit{cannot} define a function of the following type in general
\[ \forall p\of{Y \to \Prop}.\; (\exists y.\; p y) \to \Sigma y \of Y.\; p y.\]

However, such an elemination of $\exists$ into $\Sigma$ is \emph{admissible} in \theCIC.
This means that any concretely given, fully constructive proof of $\exists y.\;p y$ without assumptions can always be given as a proof of $\Sigma y.\; p y$.
Note that admissibility of a statement is strictly weaker than provability, and in general does not even entail its consistency.

Crucially, \theCIC allows defining large elimination principles for the falsity proposition $\bot$ and for equality.
Additionally,
for some restricted types $Y$ and restricted predicates $p$, one \textit{can} define a large elimination principle for existential quantification.
In particular, this holds for $Y = \nat$ and $p (n \of \nat) := f n = \btrue$ for a function $f\of{\nat\to\bool}$.

\begin{fact}[][computational_explosion]
  One can define functions of type
  \begin{enumerate}
  \item $\forall A \of \Type.\; \bot \to A$
  \item $\forall X \of \Type. \forall A \of X \to \Type.\forall x_1 x_2 \of X.\;x_1 = x_2 \to A x_1 \to A x_2$.
  \item $\forall f \of{ \nat \to \bool}.\;(\exists n.\; f n = \btrue) \to \Sigma n .\; f n = \btrue$
  \end{enumerate}
\end{fact}

We will not need any other large elimination principle in this paper.
A restriction of large elimination in general is necessary for consistency of Coq~\cite{coquand:inria-00075471}.
As a by-product, the computational universe $\Type$ is separated from the logical universe $\Prop$, allowing classical logic in $\Prop$ to be assumed while the computational intuitions for $\Type$ remain intact.

The intricate interplay between $\Sigma$ and $\exists$ is in direct correspondence to the status of the axiom of choice in \theCIC.
The axiom of choice was first stated for set theory by Cantor.
In the formulation by Cantor, it is equivalent to the statement that every total, binary relation contains the graph of a function, i.e.:
\[ \forall R \subseteq X \times Y.\;(\forall x.\exists y.(x,y) \in R) \to \exists f\of{X\to Y}.\forall x.\;(x,fx) \in R  \]

Here $\ofbox{X \to Y}$ is the set-theoretic function space.
As usual, such a classical principle can also be stated in type theory.
However, the concrete formalisation crucially depends on how the notion of a (set-theoretic) function is translated:
While in set theory the term \textit{function} is just short for \textit{functional relation},
in \theCIC functions and (total) functional relations are different objects, we thus discuss both possible translations of the axiom of choice here.

\newcommand\inhabited[1]{|\!| #1 |\!|}
The more common version, used e.g.\ by Bishop~\cite{bishop2012constructive}, is to use type-theoretic functions for set-theoretic functions, i.e.\ state the type-theoretic axiom of (functional) choice~as
\[ \forall R\of{X\to Y \to \Prop}.\; (\forall x.\exists y.\; R x y) \to \exists f\of{X\to Y}.\forall x.\;R x (f x) \]
Since in type theory proofs are first class object, one can equivalently state a principle postulating the inhabitedness of the following type:
\[ {\forall p\of{Y \to \Prop}.\; (\exists y.\; p y) \to \Sigma y \of Y.\; p y}\]
Note how this is exactly the non-provable correspondence of $\exists$ and $\Sigma$ discussed above.
This formulation makes clears why in Martin-Löf type theory as implementation of Bishop's constructive mathematics, where one defines $\exists := \Sigma$,
the axiom of choice is accepted since it can be proved.
In the context of Church's simple type theory, this axiom is also known as axiom of indefinite description~\cite{andrews2002introduction}.

\subsection{Notions of synthetic computability}
\enlargethispage{2\baselineskip}

We call a predicate  $p \of X \to \Prop$ \dots
\begin{itemize}
\item \textit{decidable} if $\decidable p := \exists f \of X \to \bool.\forall x \of X.\; p x \leftrightarrow f x = \btrue$.
\item \textit{enumerable} if $\enumerable p := \exists f \of \nat \to \option X.\;\forall x \of X.\;p x \leftrightarrow \exists n \of \nat.\;f n = \Some x$.
\item \textit{semi-decidable} if $\semidecidable p := \exists f \of {X \to \nat \to \bool}.\;p x \leftrightarrow \exists n.\;f x n = \btrue$.
\end{itemize}

A type $X$ is discrete if $\lambda (x_1, x_2) \of X \times X.\;x_1 = x_2$ is decidable and
enumerable if $\lambda x : X.\;\top$ is enumerable.
We repeat the following facts from~\cite{forster2020churchs}:

\setCoqFilename{Synthetic.SemiDecidabilityFacts}%
\begin{fact}[][decidable_semi_decidable]
  \label{coq:semi_decidable_enumerable}
  \label{coq:decidable_semi_decidable}
  The following hold:
  \begin{enumerate}
    \coqitem[decidable_semi_decidable] Decidable predicates are semi-decidable and co-semi-decidable.
    \setCoqFilename{Synthetic.EnumerabilityFacts}%
    \coqitem[semi_decidable_enumerable] Semi-decidable predicates on enumerable types are enumerable.
    \coqitem[enumerable_semi_decidable] Enumerable predicates on discrete types are semi-decidable.
    \setCoqFilename{Synthetic.SemiDecidabilityFacts}%
  \coqitem[sdec_co_sdec_comp] The complement of semi-decidable predicates is co-semi-decidable.
  \end{enumerate}
\end{fact}

\setCoqFilename{Synthetic.DecidabilityFacts}
\begin{fact}[][dec_compl]
  Decidable predicates are closed under complements.
  Decidable, enumerable, and semi-decidable predicates are closed under conjunction and disjunction.
\end{fact}

One can also characterise the notions via partial functions:
\setCoqFilename{Synthetic.EnumerabilityFacts}
\begin{fact}\label{coq:semi_decidable_part_iff}
  \begin{enumerate}
    \coqitem[enumerable_part_iff]
    $\enumerable p \leftrightarrow \exists f \of{ \nat \pfun X}.\forall x.\; p x \leftrightarrow \exists n. \; f n \hasvalue x$
    \setCoqFilename{Synthetic.SemiDecidabilityFacts}
    \coqitem[semi_decidable_part_iff]
    $\semidecidable p \leftrightarrow \exists Y (f \of{ X \pfun Y}).\forall x.\;p x \leftrightarrow \exists y.\;f x \hasvalue y$
  \end{enumerate}
\end{fact}

Lastly, we introduce many-one reducibility.
A predicate $p \of X \to \Prop$ is many-one reducible to a predicate $q \of Y \to \Prop$ if
\[ p \redm q := \exists f \of {X \to Y}.\;\forall x.\;p x \leftrightarrow q (f x). \]

\begin{fact}\label{coq:enumerable_red}
  \begin{enumerate}
    \coqitem[decidable_red] If $p \redm q$ and $q$ is decidable, then $p$ is decidable.
  \coqitem[enumerable_red]
  Let $X$ be enumerable, $Y$ discrete, and $p \of{X \to \Prop}$, $q \of{ Y \to \Prop}$.
  If $p \redm q$ and $q$ is enumerable then $p$ is enumerable.
\end{enumerate}
\end{fact}

\section{Church's thesis}
\label{sec:CTdef}\label{def:CT}

Textbooks on computability start by defining a model of computation, Rogers \cite{rogers1987theory} uses $\mu$-recursive functions.
As center of the theory, Rogers defines a step-indexed interpreter $\psi$ of all $\mu$-recursive functions.
An application $\psi_c^n x$ denotes executing the $\mu$-recursive function with code $c$ on input $x$ for $n$ steps.

Once some evidence is gathered, Rogers (as well as other authors) introduce the Church-Turing thesis, stating that all intuitively calculable functions are $\mu$-recursively computable.
Using the Church-Turing thesis, $\psi$ has the following (informal) universal property:
\[ \forall f \of \nat \to \nat.\;\textit{ intuitively computable } f \to \exists c \of \nat.\forall x \of \nat.\exists n.\;\psi_c^n x = \Some (f x) \]

Note that the property is really only pseudo-formal:
The notion of intuitive calculability is not made precise, which is exactly what allows $\psi$ to stay abstract for most of the development.
Every invocation of the universality could be replaced by an individual construction of a ($\mu$-recursive) program, but relying solely on the notion of intuitive calculability allows Rogers to build a theory based on a function $\psi$ which could equivalently be implemented in any other model of computation.
Since not every function $f \of{\nat\to\nat}$ in the classical set theory Rogers works in\footnote{\enquote{We use the rules and conventions of classical
    two-valued logic (as is the common practice in other parts of mathematics), and we say
    that an object exists if its existence can be demonstrated within standard set theory.
    We include the axiom of choice as a principle of our set theory.}~\cite[pg. 10, footnote $\dagger$]{rogers1987theory}}
is intuitively computable, every invocation of the universality of $\psi$ has to be checked individually to ensure that it is indeed for an intuitively calculable function.

We however do not work in classical set theory, but in \theCIC, a constructive system.
As in all constructive systems, every definable function is %
intuitively calculable.
It is thus natural to assume that the universal function $\psi$ is universal for \textit{all} functions $f\of{\nat\to\nat}$.
For historical reasons, this axiom is called $\CT$ (\enquote{Church's thesis})~\citetwo{kreisel1965mathematical}{troelstra1988constructivism}.

We define $\CT_{\psi}$ parametric in a step-indexed interpreter $\psi\of{\nat\to\nat\to\nat\to\option \nat}$.
As before, we write an evaluation of code $c$ on input $x$ for $n$ steps as $\psi_c^n x$ instead of $\psi\, c\, x\, n$.
For step-indexed interpreters, the sequence $\lambda n.\psi_c^n x$ is always stationary:
\[ \forall c x n_1 v.\;\psi_c^{n_1} x = \Some v \to \forall n_2.\; n_2 \geq n_1 \to \psi_c^{n_2} x = \Some v \]

Now $\CT_{\psi}$ states that $\psi\of{\nat\to\nat\to\nat\to\option\nat}$ is universal for \textit{all} functions $f\of{\nat\to\nat}$:
\[ \CT_{\psi} := \forall f \of {\nat \to \nat}.\exists c\of\nat.\forall x\of\nat.\exists n\of\nat.\; \psi_c^n x = \Some (f x)  \]

One can also see $\phi$ as an enumeration of stationary functions from $\nat$ to $\nat$,
which enumerates every total function $f$.

$\CT_{\psi}$ is not provable in \theCIC, independent of the definition of $\psi$. %
However, when $\psi$ is a step-indexed interpreter for a model of computation,
$\CT_{\psi}$ is the well-known constructivist axiom $\CT$, see \cite{forster2020churchs} for a treatment of its status in \CIC.
We give an overview over arguments why $\CT$ is consistent in \CIC in \Cref{sec:CTcons}.

In contrast to textbook proofs, proofs of theorems based on $\CT_{\psi}$ do not have to be individually checked for valid applications of the Church-Turing thesis.

As stated above, $\CT_{\psi}$ applies to unary total functions, but is immediately extensible to $n$-ary functions $f \of{ \nat^n \to \nat}$ using pairing. %
Partial application for such $n$-ary functions is realised via the $S^m_n$ theorem.
We only state the case $m=n=1$, which implies the general case:
\label{def:SMN}%
\begin{align*}
  \SMN_\psi &:= \Sigma \sigma \of{\nat\to\nat\to\nat}.\forall c x y v.\; (\exists n.\;\psi_{\sigma c x}^n y = \Some v) \leftrightarrow (\exists n.\;\psi_c^n \langle x,y \rangle = \Some v) %
\end{align*}

Note that we formulate $\SMN$ with a $\Sigma$ rather than an $\exists$.
For the results we consider in this paper, the different is largely cosmetic.
The formulation with $\Sigma$ allows the construction of functions accessing $\sigma$ directly, rather than only being able to prove the existence of functions based on $\sigma$.

The key property of $\CT_\psi$ is that it allows the definition of an enumerable but undecidable problem:

\setCoqFilename{Axioms.bestaxioms}
\begin{lemma}[][CT_halting]
  Let $\phi$ be stationary. Then $\CT_\psi \to \Sigma p \of{\nat\to\Prop}.\;\semidecidable p \land \neg \semidecidable \compl p \land \neg \decidable p \land \neg \decidable \compl p$.
\end{lemma}
\begin{proofqed}
  One can define $p c := \exists n m.\;\psi_c^m \langle c,n \rangle = \Some 0$
  and clearly we have $\semidecidable p$.
  If $f \of {\nat\to\nat\to\bool}$ is a semi-decider for $\compl p$,
  let $c$ be its code w.r.t.\ $\psi$.
  Then $p c \leftrightarrow \neg p c$, contradiction.
  Thus $p$ is also not decidable.
\end{proofqed}

\section{Synthetic Church's Thesis}
\label{sec:SCT}

By keeping $\psi$ abstract and assuming $\CT_{\psi}$, one never has to deal with encodings in a model of computation.
However, formal proofs involving the $\SMN_{\psi}$ axiom are tedious.
We identify the axiom \introterm{synthetic Church's thesis} $\SCT$ as a more convenient variant of $\CT_{\psi} \land \SMN_\psi$,
which postulates a step-indexed interpreter $\psi$ \introterm{parametrically universal} for $\nat\to\nat$:
\label{def:SCT}
\begin{align*}
  \SCT := \Sigma \psi &\of{\nat \to \nat \to \nat \to \option \nat}. \\
                       &(\forall c\, x\, n_1\, n_2\, v.\;\psi_c^{n_1} x = \Some v \to n_2 \geq n_1 \to \psi_c^{n_2} x = \Some v) ~\land \\
                       &\forall f \of{\nat \to \nat \to \nat}.\;\exists \gamma \of \nat \to \nat.\; \forall i x.\exists n.\; \psi_{\gamma i}^{n} x = \Some (f_i x)
\end{align*}

By parametrically universal we mean that for any family of functions $f_i \of{\nat\to\nat}$ parameterised by $i \of \nat$,
we obtain a coding function $\gamma$ s.t.\ $\gamma i$ is the code of $f_i$, i.e.\ $\psi_{\gamma i}$ agrees with~$f_i$.

The consistency of $\SCT$ follows from the consistency of $\CT$ formulated for a Turing-complete model of computation.
For this purpose, we choose the weak call-by-value $\lambda$-calculus $\L$, which we discuss in detail in \Cref{sec:CTinL}.
Conversely, one can recover non-parametric universality of $\psi$ from parametric universality:

\begin{theorem}[][CT_SMN_to_SCT]
  Let $\psi_{\L}$ be a step-indexed interpreter for $\L$.
  For any $\psi$ \st\ $\lambda n.\psi_c^n x$ is stationary we have the following:
  \begin{enumerate}
  \item $\CT_{\psi_{\L}} \to \Sigma \psi.\;\CT_\psi \land \SMN_\psi$
  \coqitem[CT_SMN_to_SCT] $\CT_\psi \to \SMN_\psi \to \SCTS$
  \coqitem[SCT_to_CT] $\SCTS \to \Sigma \psi.\; \CT_\psi$
  \end{enumerate}
\end{theorem}
\begin{proofqed}
  (1) follows by proving $\SMN_{\psi_{\L}}$, which we do in \Cref{sec:CTinL}, see \Cref{lem:CTL_to_SMN}.

  For (2),
  let $\psi$ and $\sigma$ be given.
  We prove that $\psi$ satisfies the condition in $\SCTS$.
  Let
 $f\of{\nat\to\nat\to\nat}$ be given.
  We obtain a code $c$ for $\lambda \langle x,y \rangle.\; f x y$.
  Now define $\gamma x \definedas \sigma c x$.

  (3) is trivial by turning the unary function $f\of{\nat\to\nat}$ into the (constant) family of functions $f'_x y := f y$.
  Now a coding function $\gamma$ for $f'$ allows to choose $\gamma 0$ as code for $f$.
\end{proofqed}

In Theorem \ref{coq:EPF_to_CT_SMN} we prove that $\SCT$ also implies $\Sigma \psi.\;\CT_{\psi} \land \SMN_{\psi}$.
Note that the consistency of $\CT_{\phi_\L}$ implies the consistency of $\Sigma \psi.\;\CT_\psi \land \SMN_\psi$ and $\SCT$.

\section{Variations of Synthetic Church's Thesis}
\label{sec:EPF}\label{def:EPF}

We have defined $\SCT$ to postulate a step-indexed interpreter $\psi \of{\nat \to (\nat \to \nat \to \option \nat)}$,
parametrically universal for $\nat\to\nat$.
In this section, we develop equivalent variations of $\SCT$.
There are three obvious points where $\SCT$ can be modified.

\begin{enumerate}
\item The return type of $\psi$ can be stationary functions of type $\nat\to(\nat\to\option\nat)$ or $\nat\to(\nat\to\option\bool)$, or partial functions of type $\nat\pfun\nat$ or $\nat\pfun\bool$, 
\item $\psi$ can be postulated to be parametrically universal for $\ofbox{\nat\to\nat}$, $\ofbox{\nat\to\bool}$, $\ofbox{\nat\pfun\nat}$, $\ofbox{\nat\pfun\bool}$, or stationary functions $\nat\to(\nat\to\option\nat)$ or $\nat\to(\nat\to\option\bool)$.
\item Coding functions $\gamma$ can be existentially quantified ($\exists$),
  computably obtained ($\Sigma$), or classically existentially quantified $\neg\neg\exists$.
\end{enumerate}

For $\SCT$, the return type of $\psi$ is stationary functions $\nat\to(\nat\to\option\nat)$, $\psi$ is parametrically universal for $\nat\to\nat$,
and $\gamma$ is existentially quantified.

For (1), it is important to see that letting $\psi$ return total functions is no option, since such an enumeration is inconsistent\rlap,%
\footnote{Note that conversely an injection of $(\nat\to\nat)\to\nat$ can likely be consistently assumed~\cite{bauer2015injection}.}
even up to extensionality:

\begin{fact}[Cantor]
  There is no $e\of{\nat\to(\nat\to A)}$ \st\ $\forall f\of{\nat\to A}.\;\exists c.e c \equiv_{\nat\to A} f$ for $A = \nat$ or $A = \bool$.
\end{fact}

For (3), the variant with $\Sigma$ is consistent, but negates functional extensionality~\cite{yamada2020game}.
Variants with $\neg\neg\exists$ are often called \textit{Weak $\CT$}~\cite{mccarty1991incompleteness}, we refrain from discussing such variants in this paper.

In this section, we discuss how all other variations of $\SCT$ are equivalent, and single out three of them:
\begin{enumerate}
\item $\EPF$, the enumerability of partial functions axiom, postulating
  $\theta \of \nat \to (\nat\pfun\nat)$ parametrically universal for $\nat\pfun\nat$,
\item $\SCT_\bool$, postulating
  $\psi \of \nat \to (\nat\to(\nat\to\option\bool))$ universal for $\nat\to\bool$, and
\item $\EPF_{\bool}$, postulating
  $\theta \of \nat \to (\nat\pfun\bool)$ universal for $\nat\pfun\bool$.
\end{enumerate}

The \introterm{enumerability of partial functions axiom} $\EPF$ is defined as:
\[ \EPFintro := \exists \theta \of{\nat \to (\nat \pfun \nat)}.\forall f \of{\nat \to \nat \pfun \nat}. \exists \gamma \of{\nat\to\nat}.\forall i.\;\theta_{\gamma i} \equivwrt{\nat\scriptpfun\nat} f_i \]

Instead of seeing $\theta$ as enumeration, we can also see $\theta$ as surjection from $\nat$ to $\nat\pfun\nat$ up to $\equivwrt{\nat\nrightarrow\nat}$.
Proving that $\SCT \leftrightarrow \EPF$ amounts to showing that any implementation of partial functions is equivalent to the implementation based on stationary sequences, and that any stationary function can be encoded in a total function $\ofbox{\nat\to\nat}$ via pairing.

\begin{theorem}[][SCT_to_EPF]\label{coq:EPF_to_SCT}
  $\SCT \leftrightarrow \EPF$
\end{theorem}
\begin{proofqed}
  The direction from left to right is by observing that
  there is a function $\mathsfe{mktotal} \of{(\nat\to\nat\pfun\nat)\to\nat\to\nat\to\nat}$ \st\ $f_i x \hasvalue v \leftrightarrow \exists n.\;\mathsfe{mktotal}\,f\, i\, \langle x, n \rangle = \succN v$ using $\mathsfe{seval}$.
  We then define
  \begin{align*}
  \theta c x := &(\mu (\lambda n.\;\iteis{\psi_c^n x}{\Some (\succN v)}{\ret \btrue}{\ret\bfalse})) \\ &\bind \lambda n.\;\iteis{\psi_c^n x}{\Some (\succN v)}{\ret v}{\mathsfe{undef}}
  \end{align*}
  The direction from right to left constructs $\psi_c^n x := \mathsfe{seval}\,(\theta_c y)\, n$.
  Let $f \of{\nat \to \nat\to\nat}$. Define the partial function $f'_i x := \ret (f_i x)$.
  Now a coding function $\gamma$ for $f'$ by $\EPF$ is a coding function for $f$ to establish $\SCT$.
\end{proofqed}

Instead of stating $\EPF$ as enumeration of partial functions, we can equivalently state it w.r.t.\ parameterised functional relations:

\begin{fact}[][EPF_iff]
  $\EPF$ is equivalent to the following:\small
  \[ \exists \theta \of{\nat\to(\nat\pfun\nat)}.\;\forall R \of {\nat \to \FunRel \nat \nat}.\;
    (\exists f.\forall i.\;f_i \textit{ computes } R_i) \to
    \exists \gamma.\forall i.\;\theta_{\gamma i} \textit{ computes } R_i
  \]
\end{fact}

\begin{theorem}[][EPF_to_CT_SMN]
  $\EPF \to \Sigma \psi.\;\CT_{\psi} \land \SMN_{\psi}\land \forall c\,x.\;\lambda n.\psi_c^n x \textit{ is stationary} $
\end{theorem}
\begin{proof}
  Let $\theta$ be given as in $\EPF$ and define $\psi^n_c x := \mathsfe{seval}\,(\theta_c y)\, n$, which allows proving $\CT_{\psi}$ as in Theorem \ref{coq:SCT_to_EPF}.
  Let furthermore $f_{\langle c,x \rangle} y := \theta c \langle x, y \rangle$ and $\gamma$ be a coding function for $f$ by $\EPF$.
  Define $S c x := \gamma \langle c, x  \rangle$.
  We have
  \begin{align*}
  \theta_{S c x} y \equiv \theta_{\gamma \langle c,x \rangle} y \equiv \theta c \langle x,y \rangle  \tag*{\qed}
  \end{align*}
\end{proof}

We introduce $\SCT_{\bool}$, postulating a step-indexed interpreter parametrically universal for $\ofbox{\nat\to\bool}$:
\begin{align*}
  \SCT_{\bool} := \Sigma \psi &\of{\nat \to \nat \to \nat \to \option \bool}. \\
                       &(\forall c x n_1 n_2 v.\;\psi_c^{n_1} x = \Some v \to n_2 \geq n_1 \to \psi_c^{n_2} x = \Some v) ~\land \\
                       &\forall f \of{\nat\to\nat \to \bool}.\;\exists \gamma.\forall i x.\exists n.\; \psi_{\gamma i}^{n} x = \Some (f_i x)
\end{align*}

$\SCT_{\bool}$ is equivalent to $\SCT$.
One direction is immediate since $\bool$ is a retract of $\nat$ (i.e.\ can be injectively embedded).
The other direction follows by mapping the infinite sequence $f 0, f 1, f 2, \dots$ to the sequence
\[ \bfalse^{f 0} \, \btrue \, \bfalse^{f 1} \, \btrue \, \bfalse^{f 2} \, \btrue \dots \]

\begin{theorem}[][SCT_bool_to_SCT]
  $\SCT_{\bool} \leftrightarrow \SCT$
\end{theorem}

We define the \introterm{parametric enumerability of partial boolean functions} axiom 
\[\EPF_{\bool} := \Sigma \theta \of{\nat \to (\nat \pfun \bool)}.\forall f \of{\nat\to\nat \pfun \bool}. \exists \gamma \of \nat \to \nat.\forall i.\;\theta_{\gamma i} \equivwrt{\nat \scriptpfun \bool} f_i \]

Recall that $\theta_{\gamma i} \equivwrt{\nat \scriptpfun \bool} f_i$ if and only if $\forall x v.\;\theta_{\gamma i} x \hasvalue v \leftrightarrow f_i x \hasvalue v$.
Proving $\EPF_{\bool}$ equivalent to $\SCT$ is easiest done by proving the following:

\begin{theorem}[][EPF_bool_to_SCT_bool]
  $\EPF_{\bool} \leftrightarrow \SCT_\bool$
\end{theorem}
\begin{proofqed}
  Exactly as in Theorem \ref{coq:EPF_to_SCT}.
\end{proofqed}

Using $\EPF_{\bool}$ it is easy to establish an enumerable, undecidable problem:

\begin{fact}[][EPF_halting]
  $\EPF_{\bool} \to \Sigma p \of{\nat\to\Prop}.\;\enumerable p \land \neg \enumerable \compl p \land \neg \decidable p$ 
\end{fact}
\begin{proof}
  Let $\theta$ be given as in $\EPF$.
  Define $\K c := \exists v. \theta_c c \hasvalue v$.
  $\K$ is semi-decided by $\lambda c n. \iteis{\mathsfe{seval} (\theta_c c) n}{\Some v}{\btrue}{\bfalse}$ and thus enumerable by Fact \ref{coq:semi_decidable_enumerable}.
  
  We prove that $\compl {\K}$ is not semi-decidable, yielding both $\neg\enumerable{\compl {\K}}$ and $\neg \decidable \K$ by Fact \ref{coq:semi_decidable_enumerable}.
  Let $\compl {\K}$ be semi-decidable, i.e.\ by Fact \ref{coq:semi_decidable_part_iff} (2) there is $f \of \nat\pfun Y$ \st\ $\neg \K x \leftrightarrow \exists y.\; f x \hasvalue y$.
  Define~$f'\of{\nat\pfun\bool}$ as $f' x := f x \bind \lambda \_.\; \ret \btrue$.
  Now $f'$ has a code $c$ \st\ $\forall x.\;e c x \hasvalue f' x$ by universality of $\theta$.

  We have a contradiction via
  \begin{align*}
    \neg \K c  \leftrightarrow (\exists y.\; f c \hasvalue y) \leftrightarrow (\exists y.\;f' x \hasvalue y) (\exists y. e c c \hasvalue y) \leftrightarrow \K c. \tag*{\qed}
  \end{align*}
\end{proof}

\section{The Enumerability Axiom}
\label{sec:EA}\label{sec:EAS}

Using $\EPF$ or $\SCT$ as basis for synthetic computability requires the manipulation of partial functions or stationary functions, which is tedious.
Alternatively, 
synthetic computability can be presented even more elegantly by an equivalent axiom concerned with enumerable predicates rather than partial functions.
A non-parametric enumerability axiom is used by Bauer~\cite{BauerSyntCT} together with countable choice to develop synthetic computability results.

We introduce the \introterm{parametric enumerability axiom} postulating an enumerator $\varphi \of \nat\to(\nat\to\option\nat)$
which is parametrically universal for all parametrically enumerable predicates~$p \of {\nat\to\nat\to\Prop}$:%
\small
\label{def:EA}\label{introterm:EAS}%
  \begin{align*}
    \EAS := &\Sigma \varphi\of{\nat \to (\nat \to \option\nat)}.\forall (p \of{\nat\to\nat \to \Prop}). \\
            &\quad(\exists (f \of \nat\to\nat\to\option\nat).\forall i.\;f_i \textit{ enumerates } p_i) \to \exists \gamma \of \nat\to\nat.\forall i.\; \varphi_{\gamma i} \textit{ enumerates } p_i
\end{align*}
\normalsize

That is, $\EA$ states that whenever $p$ is parametrically enumerable,
then $\lambda i.\;\varphi_{\gamma i}$ parametrically enumerates $p$ for some $\gamma$.

Note the two different roles of natural numbers in the axiom:
If we would consider predicates over a general type $X$ we would have $\varphi \of{\nat \to (\nat \to \option X)}$.

Equivalently, we could have required that $p$ is enumerable:

\begin{fact}[][EA_iff_enumerable]
  $\EAP$ is equivalent to \[ \Sigma \varphi.\;\forall p \of{\nat\to\nat\to\Prop}.\;\enumerable (\lambda (x,y).\;p (x,y)) \to \exists \gamma \of{\nat\to\nat}.\forall i.\;\varphi_{\gamma i} \textit{ enumerates } p i.\]
\end{fact}

Again equivalently, $\EA$ can be stated to only mention enumerators instead of predicates, which is the formulation of $\EA$ used in~\cite{forster2020churchs}.
\begin{fact}[][EA_ran_iff]\label{lem:EA_ran_iff}
  $\EA \leftrightarrow \Sigma \varphi \of{\nat \to (\nat \to \option\nat)}.\forall f \of{\nat \to \nat \to \option \nat}. \exists \gamma.\forall x.\; \varphi_{\gamma x} \equiv_{\textsf{ran}} f x$
\end{fact}

In this formulation, $\varphi$ is a surjection w.r.t.\ range equivalence $f \equiv_{\textsf{ran}} g$, where $\varphi_c \equiv_{\textsf{ran}} f \leftrightarrow \forall x.(\exists n.\varphi_c n = \Some x) \leftrightarrow (\exists n. f n = \Some x)$.

\newcommand\W{\mathcal{W}}
Given $\varphi$, we define $\W_c x := \exists n.\; \varphi_c n = \Some x$
and the problem $\K$ as the diagonal of $\W$, i.e.\ $\K c := \W_c c$.
We call $\W$ a \introterm{universal table}.
One can show that $\W$ and $\K$ are $m$-equivalent, and both are $m$-complete.
For now we only use $\K$ to note the following result:

\begin{lemma}[][EA_halting]
  $\EA \to \Sigma p \of \nat \to \Prop.\;\enumerable p \land \neg \enumerable \compl p \land \neg \decidable p \land \neg \decidable \compl p$
\end{lemma}
\begin{proofqed}
  We pick $p$ as $\K c := \W_c c$.
  $\K$ is enumerated by \[\lambda \langle c,m \rangle.\;\iteis{\varphi_c m}{\Some x}{\ite{x =_{\bool} c}{\Some c}{\None}}{\None}.\]
  If $\compl {\K}$ would be enumerable, there would be a code $c$ s.t.\ $\forall x.\;\W_c x \leftrightarrow \compl {\K} x$.
  In particular $\W_c c \leftrightarrow \compl{\K} c \leftrightarrow \neg \W_c c$.
  Contradiction. Then $ \neg \decidable p \land\neg \decidable \compl p$ is trivial.
\end{proofqed}

Similarly to how $\SCT$ can be reformulated
by letting $\psi$ be universal for unary functions and introducing an explicit $S^1_1$-operator, $\EA$ can also be stated in this fashion, with an $S^1_1$-operator w.r.t.\ $\W$.
\begin{lemma}[][EA_to_EA_star]
  $\EAS$ is equivalent to \[ \Sigma \varphi. (\forall p.\enumerable p \to \exists c.\varphi_c \textit{ enumerates } p) \land \Sigma \sigma\of{ \nat \to \nat \to \nat}.\forall c x y.\W_{ (\sigma c x)} y \leftrightarrow \W_c \langle x, y \rangle.\]
\end{lemma}
\begin{proofqed}
  The direction from right to left is straightforward using Lemma \ref{coq:EA_iff_enumerable}.  

  For the direction from left to right, let $\varphi$ be given.
  
  For the first part of the conclusion
  let $p$ be given and enumerable.
  Then $\lambda x y.\; p y$ is parametrically enumerable, so let $\gamma$ be given from $\EA$.
  Then $\varphi_{{ \gamma 0}}$ enumerates $p$.
  For the second part, let $p \langle c,x \rangle y := \exists n. \varphi_c n = \Some \langle x,y \rangle $.
  Since $p$ is enumerable, by Lemma \ref{coq:EA_iff_enumerable} and $\EAS$ there is $\gamma$ \st $\varphi_{\gamma \langle c,x \rangle}$ enumerates $p \langle c,x \rangle$.
  Now $S c x := \gamma \langle c,x \rangle$ is the wanted function.
\end{proofqed}

$\SCT$ and $\EA$ are equivalent.
For the forwards direction, we show that enumerators $\nat\to\option\nat$ can be equivalently given as functions $\nat\to\nat$.

\begin{theorem}[][SCT_to_EA]
  $\SCT \to \EA$
\end{theorem}
\begin{proof}
  Let a universal function $\psi$ be given. Define:
  \[\varphi_c \langle n,m \rangle \definedas \iteis{\psi_c^n m}{\Some (\succN x)}{\Some x}{\None}\]
  Let $f\of{\nat\to\nat\to\option\nat}$ be a parametric enumerator for $p$.
  We define $f'\of{\nat\to\nat\to\nat}$ as $f' x n := \iteis{f x n}{\Some y}{\succN y}{0}$.
  By $\SCT$, we obtain a function $\gamma$ for $f'$, and we have $\forall x y.\;p x y \leftrightarrow \exists k.\;\varphi_{\gamma x} k = \Some y$ as wanted. \qed
\end{proof}

For the converse direction, we use that the graph of functions is enumerable.

\begin{theorem}[][EA_to_SCT]
  $\EA \to \SCT$
\end{theorem}
\begin{proofqed}
  Let $\varphi$ as in $\EA$ be given.
  Recall $\mathsf{mkstat} \of{(\nat\to\option X) \to \nat \to \option X}$ 
  turning arbitrary $F\of{\nat\to\option\nat}$ into stationary sequences.
  We define
  \small
  \begin{align*}
    \varphi_c^n x := \mathsf{mkstat} (\lambda n.\; \iteis{\varphi_c n&}{\Some \langle x',y \rangle\\&}{\ite{x' =_{\bool} x}{\Some y}{\None}}{\None}) n
  \end{align*}
  \normalsize
  Let $f\of{\nat\to\nat\to\nat}$ and
  let $\varphi_{\gamma x}$ enumerate $\lambda x \langle n,m \rangle.\;f x n = m$ via $\EA$.
  Now $\gamma$ serves as coding function for $f$ by Fact \ref{coq:mkstat_spec}.  
\end{proofqed}

\section{%
  Rice's theorem}
\label{sec:rice}
\label{sec:rice1}
\setCoqFilename{Synthetic.Rice}

One of the central results of every introduction to computability theory is Rice's theorem~\cite{Rice1953}, stating that non-trivial semantic predicates on programs are undecidable.
Two proof strategies can be found in the literature:
By using a fixed-point theorem or by establishing a many-one reduction from $\compl \K$.
We here give synthetic variants of both proofs.

We base the first proof on the axiom $\EPF$, since the notion of a fixed-point is more natural there.
We base the second proof on the axiom $\EA$. Here the choice is less canonical, but using $\EA$ enables a comparison of $\EA$ and $\EPF$ as axioms for synthetic computability.

We start by assuming $\EPF$ and proving a fixed-point theorem due to Rogers~\cite{rogers1987theory}.

\begin{theorem}[][FP]%
  Let $\theta$ be given as in $\EPF$ and $\gamma \of \nat\to\nat$, then there exists $c$ \st\ $\theta_{\gamma c} \equiv \theta_{c}$.
\end{theorem}
\begin{proofqed}
  Let $\gamma \of{\nat\to\nat}$.
  Let $f_x z := \theta_x x \bind \lambda y. \theta_y z$ and $\gamma'$ via $\EPF$ be such that $\theta_{\gamma' x} \equiv f_x$ (1).
  Let $c$ via $\EPF$ be such that $\forall x.\;\theta_c x \hasvalue \gamma (\gamma' x)$ (2).
  
  Now $f_c \equiv \theta_{\gamma' c}$ by (1).
  
  Also $f_c z \equiv (\theta_c c \bind \lambda y. \theta_y z) \equiv \theta_{\gamma(\gamma' x)} z$ by the definition of $f$ and (2).
 
  Now $\gamma' c$ is a fixed-point for $\lambda i.\theta_{\gamma i}$: $\theta_{\gamma (\gamma' c)} \equiv f_c \equiv \theta_{\gamma' c}$.
\end{proofqed}

Rice's theorem can then be stated and proved as follows:

\begin{theorem}[][Rice_Theorem']\label{thm:rice1}
  Let $\theta$ be given as in $\EPF$ and $p\of{\nat\to\Prop}$.
  If $p$ treats elements as codes w.r.t.\ $\theta$ and is non-trivial, then $p$ is undecidable.
  Formally:
  \[ (\forall c c'.\;\theta_{c} \equiv \theta_{c'} \to p c \leftrightarrow p c') \to \forall c_1 c_2.\; p c_1 \to \neg p c_2 \to \neg \decidable p \] 
\end{theorem}
\begin{proofqed}
  Let $f$ decide $p$ and let $p c_1$ and $\neg p c_2$.
  Define $h \of{\nat\to\nat\pfun\nat}$ as
  $h_x := \ite{f x}{\theta_{c_2}}{\theta_{c_1}}$ and let $\gamma$ via $\EPF$ be such that $\theta_{\gamma x} y \equiv h_x y$.
  Let $c$ be a fixed-point for~$\gamma$ via Theorem \ref{coq:FP}, i.e.\ $\theta_{\gamma c} \equiv \gamma$.

  Then either $f c = \btrue$ and thus $p c$, but $\theta_c \equiv \theta_{c_2}$ and thus $p c_2$. A contradiction.

  Or $f c = \bfalse$ and thus $\neg p c$, but $\theta_c \equiv \theta_{c_1}$ and thus $\neg p c_1$. A contradiction.  
\end{proofqed}

Rice's theorem is often also stated for predicates $p \of (\nat\pfun\nat) \to \nat$.
This formulation has the advantage that the requirement on $p$ does not have to mention~$\theta$.

\begin{corollary}[][Rice_HO']
  $\EPF$ implies that if $p \of (\nat\pfun\nat) \to \nat$ is extensional and non-trivial, then $p$ is undecidable.
  Formally:
  \[ \EPF \to (\forall f f' \of \nat \pfun \nat. f \equivwrt{\nat\scriptpfun\nat} f' \to p f \leftrightarrow p f') \to \forall f_1 f_2.\;p f_1 \to \neg p f_2 \to \neg \decidable p\]
\end{corollary}
\begin{proofqed}
  Let $p$ be decidable.
  We define the index predicate of $p$ as $I_p := \lambda c \of \nat.\;p(\theta_c)$, and have $I_p \redm p$.
  Thus since $p$ is decidable, $I_p$ is decidable.
  Since $I_p$ treats elements as codes and is non-trivial using $\EPF$, we have that $I_p$ is undecidable by Theorem \ref{thm:rice1}.
  Contradiction.
\end{proofqed}

A second proof strategy for Rice's theorem is by establishing a many-one reduction from a problem proved undecidable via diagonalisation. %
We could use $\K$ defined using $\EPF$ in Fact \ref{coq:EPF_halting},
but here use $\EA$ to compare the two axioms.
Thus, we use the problem $\K$ as used in Fact \ref{coq:EA_halting}.
We follow Forster and Smolka~\cite{forster2017weak}, who mechanise a fully constructive proof of Rice's theorem based on the call-by-value $\lambda$-calculus by isolating a reduction lemma (\enquote{Rice's Lemma}).

\begin{lemma}[][Rice]
  Let $\varphi$ be as in $\EA$
  and $p \of \nat \to \Prop$.
  If $p$ treats elements as codes w.r.t.\ $\varphi$,
  $p$ is non-trivial
  and the code for the empty predicates satisfies $p$,
  then
  $\compl \K \redm p$.
  
  Formally let  $\W_c x := \exists n.\;\varphi_c n = \Some x$.
  We then have:\small
  \[
    (\forall c c'.\; (\forall x.\;\W_{c} x \leftrightarrow \W_{c'} x) \to p c \leftrightarrow p c')
    \to \forall c_\emptyset c_0.\;(\forall x.\neg \W_{c_\emptyset} x) \to p c_{\emptyset} \to
    \neg p c_0 \to \compl \K \redm p
  \]  
\end{lemma}
\begin{proofqed}
  The predicate $q := \lambda x y.\;\K x \land \W_{c_0} y$ is enumerable, meaning we obtain
  $\gamma$ from $\EA$ \st\ $\forall x y.\;\W_{\gamma x} y \leftrightarrow \K x \land \W_{c_0} y$.

  Let $\neg \K x$.
  We have $\W_{\gamma x} y \leftrightarrow \bot \leftrightarrow \W_{c_\emptyset} y$.
  Since $p c_{\emptyset}$ and $p$ is semantic also $p (\gamma x)$.
  Conversely, let $p (\gamma x)$ and $\K x$.
  We have $\W_{\gamma x} y \leftrightarrow \W_{c_0} y$.
  Since $p$ is semantic, also $p {c_0}$. Contradiction.
\end{proofqed}

\begin{theorem}[][Rice_Theorem]
  Let $\varphi$ be given as in $\EA$
  and $p \of \nat \to \Prop$.
  If $p$ treats inputs as codes w.r.t.\ $\varphi$ and~$p$ is non-trivial,
  then $p$ is not bi-enumerable.
  Formally let  $\W_c x := \exists n.\;\varphi_c n = \Some x$ be the universal table for $\varphi$.
  We then have:
  \[
    (\forall c c'. (\forall x.\;\W_{c} x \leftrightarrow \W_{c'} x) \to p c \leftrightarrow p c')
    \to \forall c_1 c_2.\; p c_1 \to \neg p c_2 \to \neg (\enumerable p \land \enumerable \compl p)
  \]  
\end{theorem}
\begin{proofqed}
  Since $\lambda x \of \nat.\bot$ is enumerable, by $\EA$ there is $c_\emptyset$ \st\ $\forall x. \neg W_{c_\emptyset} x$.
  Now let $p c_1$, $\neg p c_2$, and let $p$ be bi-enumerable.
  
  If $p c_\emptyset$, we have $\compl {\K} \redm p$, a contradiction since $\compl \K$ would be enumerable by Fact~\ref{coq:enumerable_red}.
  If $\neg p c_\emptyset$ we have $\compl {\K} \redm \compl p$, again a contradiction.
\end{proofqed}

\begin{corollary}[][Rice_TheoremCorr]
  Let $\varphi$ be given as in $\EA$
  and $p \of \nat \to \Prop$.
  If $p$ treats inputs as codes w.r.t.\ $\varphi$ and $p$ is non-trivial,
  then $p$ is undecidable. \label{thm:ricecorr}
\end{corollary}

We can state this second version of Rice's theorem for $p \of (\nat\to\Prop)\to\Prop$.
\begin{corollary}[][Rice_HO]
  $\EA$ implies that if $p$ is extensional and non-trivial w.r.t.\ enumerable predicates, then $p$ is undecidable.
  Formally under the assumption of $\EA$ we have \small
  \[ (\forall q q' \of \nat \to \Prop. (\forall x.\;q x \leftrightarrow q' x) \to p q \leftrightarrow p q') \to \forall q_1 q_2.\;\enumerable q_1 \to \enumerable q_1 \to p q_1 \to \neg p q_2 \to \neg \decidable p\]
\end{corollary}
\begin{proofqed}
  Let $p$ be decidable.
  We define the index predicate of $p$ as $I_p := \lambda c \of \nat.\;p(\W_c)$, and have $I_p \redm p$.
  Thus since $p$ is decidable, $I_p$ is decidable.
  Since $I_p$ treats elements as codes and is non-trivial using $\EA$, we have that $I_p$ is undecidable by Theorem~\ref{thm:ricecorr}.
  Contradiction.
\end{proofqed}

We have formulated both theorems to explicitly assume $\theta$  and $\varphi$ and their respective specification, to contrast the axioms $\EPF$ and $\EA$.
One can however obtain Theorem \ref{thm:rice1} from Theorem \ref{thm:ricecorr} -- and vice versa --
constructing $\theta$ from $\varphi$ and constructing a predicate $q$ treating elements as indices w.r.t.\ $\theta$ from a predicate $p$ treating elements as indices w.r.t.\ $\varphi$ -- and vice versa.

Proofs based on $\EPF$ require the manipulation of partial functions, which is formally tedious.
We will thus use $\EA$ as basis for synthetic computability:
In contrast to $\SCT$, it does not force us to encode every computation as total function $\nat\to\nat$,
and in contrast to $\EPF$ it does not force us to work with partial functions either.

Instead, we can simply consider enumerable predicates
and their enumerators, which are total functions.

\section{$\CT$ in the weak call-by-value $\lambda$-calculus}
\label{sec:CTinL}
\enlargethispage{-1\baselineskip}

\sfcommand{tm}
\sfcommand{closed}
\sfcommand{embed}
\sfcommand{unembed}
\newcommand{\encode}{\varepsilon}
\sfcommand{eval}
\sfcommand{unenc}

In this section we treat $\CT_{\L}$, the formulation of $\CT_{\phi}$ where $\phi$ is a universal function for the weak call-by-value $\lambda$-calculus $\L$~\cite{plotkin1975call,forster2017weak}.
We largely omit technical details in this section and refer the interested reader to the accompanying Coq code.
We only need a stationary function $\eval \of {\tm_{\L} \to \nat \to \option(\tm_\L)}$
such that $\exists n.\l\eval\;s\;n=\Some t$ for a $\lambda$-term $s$ if and only if $t$ is the normal form of $s$ w.r.t.\ weak call-by-value evaluation,
and a function $\encode_{\nat} \of {\nat \to \tm_{\L}}$ encoding natural numbers as terms, e.g.\ using Scott encoding~\cite{scott1968}, and an inverse function $\unenc$.

We define $\phi$ such that for an application $\phi_c^n x$ we translate the code $c$ to a term $t$
and then evaluate the application $t ~ (\encode_{\nat} x)$ for $n$ steps using the step-indexed interpreter $\eval$.
To interpret $c$ as term, we need the following:

\begin{fact}\label{coq:enum_closed_proc}
  There are functions $R \of {\nat\to\option(\tm_{\L})}$ and $I \of {\tm_{\L} \to \nat}$ \st\ we have $\closed\;s  \leftrightarrow R (I s) = \Some s$. %
\end{fact}

We can then define:
\begin{align*}
  \phi_c^nx := \iteis{R c&}{\Some t}{\\&\iteis{{\eval\;(t~(\encode_{\nat} n))\;n}}{\Some v}{\unenc\; v}{\None}}{\None}
\end{align*}

\begin{fact}
  If $f \of{\nat\to\nat}$ is computed by $t$ then
  $\forall x.\exists n.\;\phi_{I t}^nx = \Some(f x)$.
\end{fact}
We write $\CT_{\L}$ instead of $\CT_{\phi}.$
To define an $S^m_n$ operator $\sigma$ we need the following:
\vspace{-1\baselineskip}
\begin{fact}\label{lem:auxiliaryembed}
  There are
  $t_{\embed}$ 
  and $t_{\unembed}$ computing $\lambda \langle n,m \rangle.(n,m)$ and $\lambda n m.\langle n,m \rangle$.
\end{fact}
We define $\sigma c x := \iteis{R c}{\Some t}{I (\lambda y.\;t (t_{\embed} (\encode_{\nat} x) y))}{c}.$
Note that here the function $I$ is applied to an $\L$-term.

\begin{theorem}[][SMN]
  $\forall c x y v.\; (\exists n.\;\phi_{\sigma c x}^n y = \Some v) \leftrightarrow (\exists n.\;\phi_c^n \langle x,y \rangle = \Some v)$
\end{theorem}
\begin{proofqed}
  We only prove the direction from right to left, the other direction is similar.
  Let $\phi_c^n \langle x,y \rangle = \Some v$.
  Then $R c = \Some t$ and $t~\encode_{\nat}(\langle x, y \rangle) \triangleright v$ for some closed term $t$.
  Let $s := (\lambda y.\; t_{\embed} (\encode_{\nat} x) y)~(\encode_{\nat} y)$.
  We have $t~\encode_{\nat}(\langle x, y \rangle) \equiv s$ and thus $s \triangleright v$.
  Thus there is $m$ s.t.\ $\eval\;m\;s = \Some v$ and we have $\phi_{\sigma c x}^m = \Some v$.
\end{proofqed}

\begin{corollary}\label{lem:CTL_to_SMN}
  $\CT_{\L} \to \Sigma \phi.\;\CT_{\phi} \land \SMN_{\phi}$
\end{corollary}

\section{Related work}

Bauer~\cite{BauerSyntCT} develops synthetic computability based on an axiom stating that the set of enumerable sets of natural numbers is enumerable.
Translating to our type theoretic setting this yields the following axiom stating that there is an enumerator $\W$ of all enumerable predicates, up to extensionality.
\[ \EAintro' := \exists \W \of{\nat \to (\nat \to \Prop)}.\forall p \of{\nat \to \Prop}.~  
  \enumerable p \leftrightarrow \exists c.~\W_c \equivwrt{\nat \scriptto \Prop} p\]

Additionally to $\EA'$, Bauer also assumes countable choice and Markov's principle.
In general however, the assumption of countable choice makes the theory anti-classical, i.e.\ assuming $\LEM$ is inconsistent.
The interplay of axioms like $\EA'$, $\MP$, $\LEM$, and countable choice is discussed in~\cite{forster2020churchs}.
Countable choice allows extracting the enumerator for every enumerable predicate in the range of $\W$ computationally, corresponding to a non-parametric version of our axiom $\EA$.
Countable choice also can be used to prove a synthetic $S^m_n$ theorem w.r.t.\ $\W$. %

Our parametric formulation of $\EA$ implies $\EA'$, and conversely $\EA'$ implies $\EA$ under the presence of countable choice.

Richman~\cite{richman1983church} introduces the axiom $\mathsfe{CPF}$ (\enquote{Countability of Partial Functions}).
It states that the set of partial functions is (extensionally) countable, i.e.\ there is a surjection $\nat\to(\nat\pfun\nat)$ w.r.t.\ equivalence on partial functions.
Intensionally, Richman models the partial function space $\nat\pfun\nat$ as stationary functions.
Thus, written out fully his axiom is a non-parametric version of $\EPF$, just instantiated to the stationary functions model of partial functions.

Theory based on $\mathsfe{CPF}$ is developed in the book by Bridges and Richman~\cite{bridges1987varieties}, where $\mathsfe{CPF}$ is taken as basis for the constructivist system $\mathsfe{RUSS}$.
In $\mathsfe{RUSS}$, the axiom of countable choice is also assumed.
Bridges and Richman discuss that \enquote{in \mathsfe{RUSS} countable choice can usually be avoided}~\cite[p. 54]{bridges1987varieties} by postulating a composition operator for $\theta$ or, equivalently, an $\SMN$ operator.

Our two proofs of Rice's theorem are in strong support of this conjecture.
Recall that the first proof is based on a parametrically universal partial function $\theta$,
while the second proof is based on a parametrically universal enumerator $\varphi$.
The two proofs also use different proof strategies.

The second strategy establishes a reduction from $\K$.
This strategy is used in the textbooks by
Cutland~\cite{cutland1980computability},
Odifreddi~\cite{odifreddi1992classical},
Soare~\cite{soare1999recursively},
and
Cooper~\cite{cooper2003computability},
whereas
Rogers~\cite{rogers1987theory}
and
Sipser~\cite{sipser2006introduction}
pose Rice's theorem as an exercise.

The first strategy, based on Rogers' fixed-point theorem or equivalently based on Kleene's recursion theorem is less frequently found.
It is however mentioned in the Wikipedia article on Rice's theorem~\cite{enwiki:1017713534}.
The technique appears first in the lecture notes by Scott~\cite{scott1968}, who shows a variant of Rice's theorem for the $\lambda$-calculus.
Scott's proof can also be found in \cite{smullyan1994,barendregt2013lambda}.

We are aware of five machine-checked proofs of Rice's theorem:
Norrish~\cite{Norrish2011} proves Rice's theorem for the $\lambda$-calculus,
formulated for predicates $p \of (\nat\to\option\nat) \to \Prop$, using the proof strategy via reduction.
Forster and Smolka~\cite{forster2017weak} prove Rice's theorem for the weak call-by-value $\lambda$-calculus,
formulated for predicates on terms of the considered calculus which have the same extensional behaviour, using the proof strategy via reduction.
Forster~\cite{forster2014} proves Scott's variant of Rice's theorem for the weak call-by-value $\lambda$-calculus,
formulated for predicates on terms of the considered calculus which do not distinguish $\beta$-equivalent terms, using a fixed-point theorem.
Carneiro~\cite{carneiro:LIPIcs:2019:11067} proves Rice's theorem for $\mu$-recursive functions,
formulated for predicates $p \of (\nat\pfun\nat) \to \Prop$, using a fixed-point theorem.
Ramos et al.~\cite{ramosformalization} prove Rice's theorem for the functional language \texttt{PVS0},
formulated for predicates on \texttt{PVS0}, using an assumed fixed-point theorem. %

Bauer~\cite{BauerSyntCT} also presents a synthetic variant of Rice's theorem.
His formulation reads
\enquote{If $A$ is a set such that all functions of type ${A \to A}$ have a fixed-point, every function ${A \to \bool}$ is constant}
and uses the enumerability axiom as discussed above, but does not rely on countable choice to the best of our knowledge.
Note that our variants of Rice's theorem presented in this paper 
are trivialities in classical set theory, the foundation of textbook computability, since both $\EPF$ and $\EA$ are false in classical set theory where all problems have a characteristing decision function.
In contrast, Bauer's theorem is a triviality in classical set theory in \textit{two} ways:
First, the enumerability axiom is contradictory in classical set theory,
and second the statement of the theorem is a trivial even without axioms since if all functions $A \to A$ have a fixed-point, $A$ is a sub-singleton:
two distinct elements $a_1,a_2$ would allow constructing a fixed-point free function $\lambda x.\ite{x = a_1}{a_2}{a_1}$.
In short, we sacrifice identifying the minimal essence of theorems to better preserve classical intuitions.

\subsection*{Acknowledgements}

I would like to thank Gert Smolka, Dominik Kirst, and Dominique Larchey-Wendling for constructive feedback and productive discussions, as well as Andrej Bauer and Dominik Wehr for helpful advice on consistency proofs of $\CT$ in the literature.

\appendix

\section{Consistency and admissibility of $\textsf{CT}$ in \CIC}
\label{sec:CT}
\label{sec:CTcons}

In 1943, Kleene conjectured that whenever $\forall x.\exists y.\;R x y$ is constructively provable, there in fact exists a $\mu$-recursive function $f$ \st\ $\forall x.\;R x (f x)$~\cite{kleene1943recursive}.
This corresponds to a strong form of the admissibility of $\CT$.
In 1945, Kleene \cite{kleene1945interpretation} proved his conjecture for Heyting arithmetic, using number realizability.
An independent proof of this is due to~Beth~\cite{beth1948semantical}.

In this section, we use $\CT$ to denote the historical formulation of $\CT$, e.g.\ using $\mu$-recursive functions, which is however equivalent to $\CT_\L$.

Troelstra and van Dalen~\cite[\S 4.5.1 p. 204]{troelstra1988constructivism} state an even stronger result, using Gödel's Dialectia interpretation~\cite{godel1958bisher},
namely that in Heyting arithmetic $\CT$, $\MP$ and a restricted form of the independence of premise rule $\IP$ (with $P$ logically decidable) are consistent as schemes.

Odifreddi states that \enquote{for all current intuitionistic systems (not involving the concept of choice sequence) the consistency with $\CT$ has actually been established}~\cite[\S 1.8 pg. 122]{odifreddi1992classical}.
We do not discuss other systems for constructive or intuitionistic mathematics in detail.

For \theCIC, the result is not explicitly stated in the literature.
An admissibility proof of $\CT$ seems to be immediate as a consequence of Letouzey's semantics extraction theorem for Coq~\cite{LetouzeyPhd}.
Regarding a consistency proof one cannot mirror the situation in Heyting arithmetic, since a Dialectia interpretation for Coq is not available~\cite{pedrot:tel-01247085}.

However, several approaches seem to yield the result:

First, $\CT$ is consistent in intuitionistic set theory (e.g. $\mathsfe{IZF}$)~\cite{hahanyan1981consistency}, and $\mathsfe{IZF}$ can be used to model \theCIC~\cite{barras2010sets}.

Secondly, realizability models based on the first Kleene algebra prove $\CT$ consistent.
Luo constructs an $\omega$-set model for the Extended Calculus of Constructions ($\mathsfe{ECC}$, a type theory with type universes and impredicative $\Prop$, but no inductive types), where  \enquote{[t]he morphisms between $\omega$-sets are 'computable' in the sense that they are realised by partial recursive functions}~\cite[\S 6.1 pg. 118]{luo1994computation}.

Thirdly, it is well known how to build topos models of the calculus of constructions~\cite{hyland1989theory}.
The effective topos, due to~Hyland\cite{hyland1982effective}, validates $\CT$.

Fourthly, Swan and Uemura~\cite{swan2019church} give a sheaf model construction proving that $\CT$ is consistent in Martin Löf type theory,
together with propostional truncation, Markov's principle, and univalence.
It seems like the syntactic universe of propositions $\Prop$ does not hinder adapting their model construction to \theCIC.

Fivthly, Yamada~\cite{yamada2020game} gives a game-semantic proof that a $\forall f.\Sigma c$ form of $\CT$ is consistent in intensional Martin Löf type theory, settling an open question of at least 15 years~\cite{ishihara2018consistency}.
Note that this form is significantly stronger, since it allows defining a strictly intensional higher-order coding \textit{function} of type $(\nat\to\nat)\to\nat$, which is inconsistent under the assumption of functional extensionality~\cite{forster2020churchs}.
It is not obvious how to extend Yamada's proof to our $\forall f.\exists c$ formulation of $\CT$ in $\CIC$ with the impredicative universe $\Prop$.

\bibliographystyle{abbrv}
\enlargethispage{-2\baselineskip}
\bibliography{biblio}

\begin{thebibliography}{10}

\bibitem{andrews2002introduction}
P.~B. Andrews.
\newblock {\em An Introduction to Mathematical Logic and Type Theory: To Truth
  Through Proof}.
\newblock Springer Netherlands, 2002.
\newblock https://doi.org/10.1007/978-94-015-9934-4.

\bibitem{barendregt2013lambda}
H.~Barendregt, W.~Dekkers, and R.~Statman.
\newblock {\em Lambda calculus with types}.
\newblock Cambridge University Press, 2013.

\bibitem{barras2010sets}
B.~Barras.
\newblock Sets in coq, coq in sets.
\newblock {\em Journal of Formalized Reasoning}, 3(1):29--48, 2010.

\bibitem{BauerSyntCT}
A.~Bauer.
\newblock First steps in synthetic computability theory.
\newblock {\em Electronic Notes in Theoretical Computer Science}, 155:5--31,
  2006.

\bibitem{bauer2015injection}
A.~Bauer.
\newblock An injection from the baire space to natural numbers.
\newblock {\em Mathematical Structures in Computer Science}, 25(7):1484--1489,
  2015.

\bibitem{beth1948semantical}
E.~W. Beth.
\newblock Semantical considerations on intuitionistic mathematics.
\newblock {\em Journal of Symbolic Logic}, 13(3):173--173, 1948.

\bibitem{bishop2012constructive}
E.~Bishop and D.~Bridges.
\newblock {\em Constructive analysis}, volume 279.
\newblock Springer Science \& Business Media, 2012.

\bibitem{bridges1987varieties}
D.~Bridges and F.~Richman.
\newblock {\em Varieties of constructive mathematics}, volume~97.
\newblock Cambridge University Press, 1987.

\bibitem{carneiro:LIPIcs:2019:11067}
M.~Carneiro.
\newblock {Formalizing Computability Theory via Partial Recursive Functions}.
\newblock In J.~Harrison, J.~O'Leary, and A.~Tolmach, editors, {\em 10th
  International Conference on Interactive Theorem Proving (ITP 2019)}, volume
  141 of {\em Leibniz International Proceedings in Informatics (LIPIcs)}, pages
  12:1--12:17, Dagstuhl, Germany, 2019. Schloss Dagstuhl--Leibniz-Zentrum fuer
  Informatik.

\bibitem{cooper2003computability}
S.~B. Cooper.
\newblock {\em Computability theory}.
\newblock CRC Press, 2003.

\bibitem{coquand:inria-00075471}
T.~Coquand.
\newblock {Metamathematical investigations of a calculus of constructions}.
\newblock Technical Report RR-1088, {INRIA}, 1989.

\bibitem{cutland1980computability}
N.~Cutland.
\newblock {\em Computability}.
\newblock Cambridge University Press, June 1980.

\bibitem{ramosformalization}
T.~M. Ferreira~Ramos, A.~A. Almeida, and M.~Ayala-Rinc{\'o}n.
\newblock Formalization of rice’s theorem over a functional language model.
\newblock Technical report, 2020.

\bibitem{forster2014}
Y.~Forster.
\newblock {\em A Formal and Constructive Theory of Computation}.
\newblock Bachelor's thesis, Bachelor's Thesis, Saarland University, 2014.

\bibitem{forster2020churchs}
Y.~Forster.
\newblock {Church’s Thesis and Related Axioms in Coq’s Type Theory}.
\newblock In C.~Baier and J.~Goubault-Larrecq, editors, {\em 29th EACSL Annual
  Conference on Computer Science Logic (CSL 2021)}, volume 183 of {\em Leibniz
  International Proceedings in Informatics (LIPIcs)}, pages 21:1--21:19,
  Dagstuhl, Germany, 2021. Schloss Dagstuhl--Leibniz-Zentrum f{\"u}r
  Informatik.

\bibitem{forster2017weak}
Y.~Forster and G.~Smolka.
\newblock Weak call-by-value lambda calculus as a model of computation in
  {C}oq.
\newblock In M.~Ayala{-}Rinc{\'{o}}n and C.~A. Mu{\~{n}}oz, editors, {\em
  Interactive Theorem Proving - 8th International Conference, {ITP} 2017,
  Bras{\'{\i}}lia, Brazil, September 26-29, 2017, Proceedings}, volume 10499 of
  {\em Lecture Notes in Computer Science}, pages 189--206. Springer, 2017.

\bibitem{godel1958bisher}
V.~K. G\"{o}del.
\newblock \"{U}ber eine bisher noch nicht ben{\"u}tzte {E}rweiterung des
  finiten {S}tandpunktes.
\newblock {\em Dialectica}, 12(3-4):280--287, Dec. 1958.

\bibitem{hahanyan1981consistency}
V.~Hahanyan.
\newblock The consistency of some intuitionistic and constructive principles
  with a set theory.
\newblock {\em Studia Logica}, 40(3):237--248, 1981.

\bibitem{hyland1982effective}
J.~M.~E. Hyland.
\newblock The effective topos.
\newblock In {\em The L. E. J. Brouwer Centenary Symposium, Proceedings of the
  Conference held in Noordwijkerhout}, pages 165--216. Elsevier, 1982.

\bibitem{hyland1989theory}
J.~M.~E. Hyland and A.~M. Pitts.
\newblock The theory of constructions: Categorical semantics and
  topos-theoretic models.
\newblock In J.~W. Gray and A.~Scedrov, editors, {\em Categories in Computer
  Science and Logic}, volume~92 of {\em Contemporary Mathematics}, pages
  137--199. Amer.\ Math.\ Soc, Providence RI, 1989.

\bibitem{ishihara2018consistency}
H.~Ishihara, M.~E. Maietti, S.~Maschio, and T.~Streicher.
\newblock Consistency of the intensional level of the minimalist foundation
  with church's thesis and axiom of choice.
\newblock {\em Archive for Mathematical Logic}, 57(7-8):873--888, Jan. 2018.

\bibitem{kleene1943recursive}
S.~C. Kleene.
\newblock Recursive predicates and quantifiers.
\newblock {\em Transactions of the American Mathematical Society},
  53(1):41--73, 1943.

\bibitem{kleene1945interpretation}
S.~C. Kleene.
\newblock On the interpretation of intuitionistic number theory.
\newblock {\em The journal of symbolic logic}, 10(4):109--124, 1945.

\bibitem{kreisel1965mathematical}
G.~Kreisel.
\newblock Mathematical logic.
\newblock {\em Lectures in modern mathematics}, 3:95--195, 1965.

\bibitem{LetouzeyPhd}
P.~Letouzey.
\newblock {\em Programmation fonctionnelle certifi{\'e}e: l'extraction de
  programmes dans l'assistant {Coq}}.
\newblock PhD thesis, 2004.

\bibitem{luo1994computation}
Z.~Luo.
\newblock {\em Computation and Reasoning: A Type Theory for Computer Science}.
\newblock Oxford University Press, Inc., USA, 1994.

\bibitem{mccarty1991incompleteness}
D.~C. McCarty.
\newblock Incompleteness in intuitionistic metamathematics.
\newblock {\em Notre Dame journal of formal logic}, 32(3):323--358, 1991.

\bibitem{Norrish2011}
M.~Norrish.
\newblock Mechanised computability theory.
\newblock In M.~C. J.~D. van Eekelen, H.~Geuvers, J.~Schmaltz, and F.~Wiedijk,
  editors, {\em Interactive Theorem Proving - Second International Conference,
  {ITP} 2011, Berg en Dal, The Netherlands, August 22-25, 2011. Proceedings},
  volume 6898 of {\em Lecture Notes in Computer Science}, pages 297--311.
  Springer, 2011.

\bibitem{odifreddi1992classical}
P.~Odifreddi.
\newblock {\em Classical recursion theory: The theory of functions and sets of
  natural numbers}.
\newblock Elsevier, 1992.

\bibitem{paulin1993inductive}
C.~Paulin-Mohring.
\newblock Inductive definitions in the system {Coq} rules and properties.
\newblock In {\em International Conference on Typed Lambda Calculi and
  Applications}, pages 328--345. Springer, 1993.

\bibitem{pedrot:tel-01247085}
P.-M. P{\'e}drot.
\newblock {\em {A Materialist Dialectica}}.
\newblock Theses, {Paris Diderot}, 2015.

\bibitem{plotkin1975call}
G.~D. Plotkin.
\newblock Call-by-name, call-by-value and the $\lambda$-calculus.
\newblock {\em Theoretical Computer Science}, 1(2):125--159, Dec. 1975.

\bibitem{Rice1953}
H.~G. Rice.
\newblock Classes of recursively enumerable sets and their decision problems.
\newblock {\em Transactions of the American Mathematical Society},
  74(2):358--366, 1953.

\bibitem{richman1983church}
F.~Richman.
\newblock Church's thesis without tears.
\newblock {\em The Journal of symbolic logic}, 48(3):797--803, 1983.

\bibitem{rogers1987theory}
H.~Rogers.
\newblock Theory of recursive functions and effective computability.
\newblock 1987.

\bibitem{scott1968}
D.~Scott.
\newblock A system of functional abstraction.
\newblock 1968.
\newblock Lectures delivered at University of California, Berkeley, Cal.,
  1962/63. Photocopy of a preliminary version, issued by Stanford University,
  September 1963, furnished by author in 1968.

\bibitem{sipser2006introduction}
M.~Sipser.
\newblock {\em {Introduction to the Theory of Computation}}, volume~2.
\newblock Thomson Course Technology Boston, 2006.

\bibitem{smullyan1994}
R.~M. Smullyan.
\newblock {\em Diagonalization and self-reference}.
\newblock Oxford science publications. Clarendon Press, Oxford, England, 1994.

\bibitem{soare1999recursively}
R.~I. Soare.
\newblock {\em Recursively enumerable sets and degrees: A study of computable
  functions and computably generated sets}.
\newblock Springer Science \& Business Media, 1999.

\bibitem{swan2019church}
A.~Swan and T.~Uemura.
\newblock On {C}hurch's thesis in cubical assemblies.
\newblock {\em arXiv preprint arXiv:1905.03014}, 2019.

\bibitem{troelstra1988constructivism}
A.~S. Troelstra and D.~van Dalen.
\newblock Constructivism in mathematics. vol. i.
\newblock {\em Studies in Logic and the Foundations of Mathematics}, 26, 1988.

\bibitem{enwiki:1017713534}
{Wikipedia contributors}.
\newblock Rice's theorem --- {Wikipedia}{,} the free encyclopedia.
\newblock
  \url{https://en.wikipedia.org/w/index.php?title=Rice's_theorem&oldid=1017713534},
  2021.
\newblock [Online; accessed 31-May-2021].

\bibitem{yamada2020game}
N.~Yamada.
\newblock Game semantics of {M}artin-{L}\"of type theory, part iii: its
  consistency with church's thesis, 2020.

\end{thebibliography}

\end{document}